\documentclass[12pt]{amsart}
\usepackage{graphicx}
\usepackage[round]{natbib}
\usepackage{amsfonts,latexsym,amsmath,amssymb,amsthm}
\usepackage{bbm}
\usepackage{geometry}                % See geometry.pdf to learn the layout options. There are lots.
\usepackage{epstopdf}
\usepackage{setspace}

\setcitestyle{aysep={}}

\newtheorem{theorem}{Theorem}[section]

\newtheorem{proposition}[theorem]{Proposition}

\newtheorem{corollary}[theorem]{Corollary}
\newtheorem{definition}[theorem]{Definition}
\begin{document}

\title[Robust permanence in ecological equations with feedbacks]{Robust permanence for ecological equations with internal and external feedbacks}
\date{}                                          

\bibliographystyle{amnat2}

\maketitle

\begin{center}
\begin{singlespace}{
\noindent \textsc{Swati Patel\footnote{Department of Evolution and Ecology and Graduate Group in Applied Mathematics, University of California, Davis USA 95616, email: swpatel@ucdavis.edu}\footnote{Current Address: Faculty of Mathematics, University of Vienna, Austria} and Sebastian J Schreiber\footnote{Department of Evolution and Ecology and Center for Population Biology, University of California, Davis USA 95616}}
}\end{singlespace}
\end{center}

\section*{abstract}
Species experience both internal feedbacks with endogenous factors such as trait evolution and external feedbacks with exogenous factors such as weather. These feedbacks can play an important role in determining whether populations persist or communities of species coexist.  To provide a general mathematical framework for studying these effects, we develop a theorem for coexistence for ecological models accounting for internal and external feedbacks.   Specifically, we use average Lyapunov functions and Morse decompositions to develop sufficient and necessary conditions for robust permanence, a form of coexistence robust to large perturbations of the population densities and small structural perturbations of the models. We illustrate how our results can be applied to verify permanence in non-autonomous models, structured population models, including those with frequency-dependent feedbacks, and models of eco-evolutionary dynamics.  In these applications, we discuss how our results relate to previous results for models with particular types of feedbacks.

\section{Introduction}
Understanding when and how species coexist  is a fundamental problem in ecology.  Permanence theory is a mathematical formalism developed to address this problem for ecological models.  Permanence is a particular form of persistence that ensures populations will persist in the face of rare but large perturbations as well as small and frequent perturbations \citep{Schreiber2006} and hence, is an appropriate notion of coexistence for ecological systems which often experience Òvigorous shake-ups, rather than gentle stirringsÓ~\citep{Jansen1998}.  Theory for showing permanence incorporates a variety of standard approaches for characterizing and analyzing dynamical systems including topological approaches, average Lyapunov functions, and measure theoretic approaches. For a review and history on this theory and these approaches see \citep{Hutson1992, Schreiber2006, Smith2011}. Here, we develop sufficient and necessary conditions for permanence for ecological equations with feedbacks to  internal or external variables.

Biologically, many internal and external variables may provide feedbacks on the ecological dynamics of species. By internal variables, we mean factors intrinsic to the populations.   For example, this is appropriate for species structured by genotypes of an ecologically-important trait, in which selection affects the frequency of each genotype, or for traits that may change due to phenotypic plasticity. In either case, internal trait changes alter population growth and drive changes in population densities, generating a potential feedback to the trait dynamics.  Furthermore, individuals within a population may also be classified into different types (e.g. age or size classes, gender, spatial location) and this may influence their growth as well as the growth of the populations they interact with.  In particular, this population structure is important for species with life stages, between which individuals can transition, or species living in patchy landscapes, between which individuals can disperse.   By external variables, we mean dynamic variables extrinsic to the populations that influence survivorship, growth rates and reproductive rates.  For example, environmental variables such as precipitation or temperature which vary in time or the constructed habitats of ecosystem engineering species often influence these demographic rates. These internal and external variables may influence coexistence and motivate us to characterize permanence in models that account for general feedbacks with these variables.  

Permanence has been studied for general dynamical systems, abstracting beyond classical ecologicla models \citep{Hutson1984_2, Butler1986, Hutson1988, Garay1989, Hale1989}.  For example,  \citet{Garay1989} characterized permanence using Morse decompositions  and \citet{Hutson1984_2, Hutson1988}, by extending work of \citet{Hofbauer1981}, found a characterization using so-called average Lyapunov functions. Combining these approaches, \citet{Garay2003} provided sufficient conditions for robust permanence for ecological equations in the standard form

\begin{equation}\label{basic_model}
\frac{dx_i}{dt}=x_if_i(x) \hspace{3mm} i=1\dots n
\end{equation}
where $x_i$ is population densities and $f_i$ is the per-capita growth rate of population $i$. Robust permanence ensures that permanence holds following small perturbations of the per-capita growth rate equations (Schreiber 2006).  
 \citet{Garay2003} and Schreiber (2000) showed that robust permanence can be characterized in terms of the average per-capita growth rates of missing species for trajectories of \eqref{basic_model} on the extinction set. These ecological equations, however, assume that the per-capita growth rates only depend on the densities of the species, ignoring internal differences amongst individuals in the populations and external influences.  

That internal and external variation exists is indisputable; no two individuals in a population are identical and environmental conditions always vary in time.  From a modeling perspective, the ubiquity of both varying internal and external variables requires careful choice on when and how to include these variables.  In some familiar cases, feedbacks are implicitly modeled, such as in some models of interspecific competition with species competing for a limited resource \citep{Schoener1976} or predator-prey models with prey switching behavior \citep{vanBaalen2001, Hutson1984}. In other cases, feedback variables are explicitly modeled and this allows them to have their own dynamics. Hence, an important scientific goal is to determine when and how these feedbacks impact populations and communities. Some studies have examined permanence in models with specific types of internal or external feedbacks, such as for internally structured populations \citep{Hofbauer2010} or environmental variation \citep{Gatica1988, Schreiber2011b, Roth2016}. However, there is no general framework for dealing explicitly with both internal and external variables. 

In the present paper, we derive sufficient conditions for robust permanence in a general model of interacting populations with internal and external feedbacks and demonstrate how it generalizes and extends prior results of models accounting for these feedbacks.  
Our main permanence results build on average Lyapunov functions developed by  \citet{Garay2003}.  In section 2, we introduce the general model and describe our main assumptions.  Then, we state our sufficient and necessary criteria for permanence and robust permanence in sections 3 and 4, respectively.  In section 5, we apply our main theorem to three distinct models with feedbacks from the environment, population structure, and trait evolution to demonstrate its broad applicability and the importance of internal and external feedbacks on species coexistence.  

\section{Model and Terminology}\label{section:model}
We extend model (\ref{basic_model}) to incorporate internal and external feedbacks.  We suppose that $n$ populations are interacting in a community and that population $i$ has density $x_i$, with $i=1\dots n$. Interactions can include competition, predation as well as mutualisms.  For each population $i$, the per-capita growth rate, $f_i$, depends on  the densities of all the species it interacts with, as well as on another set of $m$ variables.  These $m$ variables can represent a combination of internal factors, such as the stages in a life cycle of a population, and external factors, such as temperature or another environmental variable.  Each of these factors is represented quantitatively by $y= y_1,\dots, y_m\in K \subset \mathbb{R}^m$, and can also change due to feedbacks with the population densities as well as all $m$ factors.
Altogether, the dynamics in this fairly general ecological scenario can be expressed with the differential equation model

\begin{equation}\label{eq:model}
\begin{aligned}
\frac{dx_i}{dt}&= x_i f_i (x, y) \hspace{2.7cm} i=1\dots n\\
\frac{d y_{j}}{dt} &= g_j(x, y)\hspace{3cm} j=1\dots m
\end{aligned}
\end{equation}
where $x=(x_1, \dots, x_n) \in \mathbb{R}_+^n=[0,\infty)^n$ is the vector of population densities. 
Note that both $f$ and $g$ can depend on both $x$ and $y$, capturing the potential feedback between the population densities and the other dynamic variables. The model form is quite general and can apply to a variety of types of feedbacks as illustrated in section \ref{examples}, where we apply our theorem to different biological scenarios.

Let $S=\mathbb{R}_+^n \times K$ be the state space for (\ref{eq:model}). We let $z.t$ denote the solution to (\ref{eq:model}) for initial condition $z=(x,y)\in S$.  For any set $Z\subseteq S$ and $I\subseteq \mathbb{R}_+$, let $Z.I=\{z.t|t\in I, z\in Z\}$.

We make the following standing assumptions:
\begin{description}
\item[S1] $x_if_i$ and $g_j$ are locally Lipschitz functions, and
\item[S2] there exists a compact set $Q\subseteq S$ such that $Q.[0,\infty) \subseteq Q$ and $z.t \in Q$ for $t$ sufficiently large
\end{description}

As we demonstrate in section \ref{examples}, both assumptions hold for many biological models. The first assumption ensures that solutions to (\ref{eq:model}) locally exist and are unique.  The second assumption corresponds to the biological reality that population densities do not grow without bound.

The extinction set $S_0:=\{z=(x,y) \in S | \prod_{i=1}^n x_i=0 \}$ is the set which has at least one species extinct, i.e., with density equal to zero.  Observe from model (\ref{eq:model}) that for any initial condition in $z \in S_0$, $z.t$ stays in $S_0$ for all time, capturing the ``no cats, no kittens" principle of closed ecological systems.

To use our model to identify the conditions that ensure community coexistence, we must formulate a precise notion of coexistence.  The importance of understanding coexistence in ecology has inspired many different notions of coexistence \citep{Schreiber2006}. 
Here, we use the notion of \emph{permanence}, which ensures that there is a positive population density that each species eventually stays above provided all populations are initially present.  Precisely, model (\ref{eq:model}) is \emph{permanent} if there is a $\beta>0$ such that for all $z \in S\backslash S_0$

\begin{equation}\label{permanence}
\liminf_{t\rightarrow \infty} x_i(t) \geq \beta \hspace{7mm} \text{for } i=1,2,\dots,n
\end{equation}
for all $i$, where $x_i(t)$ is the $i^{th}$ component of $z.t=(x,y).t$. 

Permanence implies that if all the species are initially coexisting, then they will continue to coexist, despite rare but large perturbations or frequent small perturbations \citep{Schreiber2006}. In the next section, we present the main theorems, which establishes sufficient and necessary conditions for permanence for models of the form (\ref{eq:model}). 

Before stating our main theorem, we introduce some terminology. The \emph{$\omega$-limit set} of a set $Z\subset S$ is $\omega(Z):=\cap_{t\geq0} \overline{Z.[t,\infty)}$ and the \emph{$\alpha$-limit set} is $\alpha(Z):=\cap_{t\leq0} \overline{Z.(-\infty,t]}$. 
A set $Z\subset S$ is \emph{invariant} if $Z.\mathbb{R} = Z$. A compact invariant set $Z$ is \emph{isolated} if there exists a closed neighborhood $U$ such that for all $z\in U\backslash Z$ there is a $t$ such that $z.t \not\in U$. For any compact invariant set $Z$, a subset $A\subset Z$ is called an \emph{attractor} in $Z$ if there is a neighborhood $U$ of $A$ such that $\omega(U\cap Z)=A$.  The dual \emph{repeller} to an attractor $A$ in $Z$ is $R(A)=\{z\in Z| \omega(z)\cap A=\emptyset\}$ and $A, R(A)$ are called attractor-repeller pairs.   

\section{Permanence theorem}\label{section:perm}
We take advantage of a characterization of permanence involving Morse decompositions.  
Roughly, a Morse decomposition of a compact invariant set is a finite number of disjoint invariant subsets, called Morse sets, ordered in such a way that the flow tends to move from sets of higher order to lower order.  More precisely,

\begin{definition}
A collection of sets $\mathcal{M}=\{M_1, M_2, ..., M_\ell\}$ is a Morse decomposition for a compact invariant set $\Gamma$ if  $M_1, M_2, ..., M_\ell$ are pairwise disjoint, isolated invariant compact sets, called Morse sets, such that for every $z\in \Gamma\backslash \cup_{k=1}^\ell M_k$ there are integers $i< j$ such that $\omega(z)\subset M_i$ and $\alpha(z)\subset M_j$. 
\end{definition}

For a compact invariant set $\Gamma$, Morse decompositions always exist but are not necessarily unique.  Trivially, one Morse decomposition of $\Gamma$ is $\{\Gamma\}$.  However, more refined Morse decompositions are typically more useful.  
In our main theorem, we use Morse decompositions to decompose the global attractor on the extinction set and define conditions on the Morse sets that give permanence for (\ref{eq:model}).

\begin{theorem}\label{maintheorem}
Let $\mathcal{M}=\{M_1, M_2, \dots M_\ell\}$ be a Morse Decomposition restricted to $S_0\cap \Gamma$ where $\Gamma$ is the global attractor for (\ref{eq:model}).  
If, for each $M_k\in \mathcal{M}$, there exists $p_{k1}, \dots, p_{kn}>0$ such that for every $z \in M_k$, there is a $T_z$ such that 
\[
\sum_{i=1}^n p_{ki} \int_0^{T_z} f_i(z.t) dt >0
\] 

then (\ref{eq:model}) is permanent.
\end{theorem}

In words, if, for each Morse set, the weighted combination (weights are $p_{ki}$) of the average per-capita growth rates over some time period is positive from every point in the Morse set, then there is permanence. 

Given the Morse decomposition restricted to  $S_0\cap \Gamma$, we can show the following partial converse

\begin{corollary}\label{main_cor3}
For each Morse set $M_k$, if there is a $p_{k1}, \dots, p_{kn}>0$ such that for every $z \in  M_k$, there is a $T_z>0$ such that
\[
\sum_i p_{ki} \int_0^{T_z} f_i(z.t) dt <0
\] 
then (\ref{eq:model}) is not permanent
and, more strongly, $S_0\cap \Gamma$ is an attractor in $\Gamma$.
\end{corollary}

The proof of this partial converse follows from applying Theorem \ref{maintheorem} to the reverse time flow in $\Gamma$ to show that $S_0\cap \Gamma$ is a repeller in reverse time.

Intuitively, for permanence to hold, population densities near extinction have to increase.
Lyapunov functions are functions that increase along solutions on some subset of the state space, and are used to characterize local and global stability of invariant sets, including equilibria.  Extending beyond equilibria, ``average Lyapunov functions", introduced by \citet{Hofbauer1981}, are functions that increase \emph{on average} along solutions.  These can apply to more complex invariant sets, which are common in many population models of the form (\ref{basic_model}). 

In Appendix \ref{mainproof}, we prove Theorem \ref{maintheorem}.  In particular, we define ``good" average Lyapunov functions (GALFs), as introduced by \citet{Garay2003}, and then prove that the existence of a GALF on each Morse set gives permanence.  Then, we show that the condition in Theorem \ref{maintheorem} on the weighted per-capita growth functions implies the existence of a GALF in the standard form

\begin{equation}\label{GALFstandard}
P(x,y)= \Pi_{i=1}^n x_i^{p_i}
\end{equation}

for some vector $p$ with $p_i>0$. The standard form (\ref{GALFstandard}) is zero everywhere on the extinction set, positive everywhere not on the extinction set and has time derivative

\[
\dot{P}=P\sum_{i=1}^n p_i f_i
\]

This time derivative gives a convenient relationship between the function and the per-capita growth rates, $f_i$, of each of the species. Feedbacks with the internal or external variable $y$ will affect the per-capita growth rates, thereby influencing the existence and construction of Lyapunov functions.

\section{Robust permanence}\label{section:robust}
Population models are always approximations of reality.  In the words of \citet{Conley1978} ``if such rough equations are to be of use, it is necessary to study them in rough terms".  In line with this, \citet{Hutson1992} introduced robust permanence, i.e. that permanence holds even with sufficiently small perturbations to the growth functions $f_i$ and \citet{Schreiber2000} subsequently provided conditions for robust permanence for (\ref{basic_model}) using a measure theoretic approach.  More recently, \citet{Garay2003} showed robust permanence for (\ref{basic_model}) using GALFs.  We use this method to extend our permanence result to robust permanence, with respect to perturbations in both the growth functions and the feedback dynamics.

Suppose we have a perturbed system 
\begin{equation}\label{eq:model_pert}
\begin{aligned}
\frac{d x_i}{dt}&= x_i \widetilde{f_i} (x, y) \\
\frac{d y_{j}}{dt} &= \widetilde{g_j}(x, y)
\end{aligned}
\end{equation}

Let $\widetilde{z.t}$ denote the solution of (\ref{eq:model_pert}) with initial condition $z \in S$ and analogously, for set $Z\subseteq S$ and $I\subseteq \mathbb{R}_+$, $\widetilde{Z.I}=\{\widetilde{z.t}| t\in I, z\in Z\}$.
Also, let $\widetilde{\omega}(Z), \widetilde{\alpha}(Z)$ denote the $\omega, \alpha$-limit set for (\ref{eq:model_pert}), respectively. Let $Q$ be as defined previously. 
We define $(\widetilde{f}, \widetilde{g})$ to be a $(\delta, Q)$-perturbation of (\ref{eq:model}) if

\begin{description}
\item[R1] $|\widetilde{f_i}(x, y)-f_i(x, y)| < \delta$ and $|\widetilde{g_j}(x, y)- g_j(x, y)|<\delta$ for all $i, j$ and for all $(x, y)\in Q$
\item[R2] $x_i\widetilde{f_i}$ and $\widetilde{g_j}$ are all locally Lipschitz continuous, and
\item[R3] $\widetilde{Q.\mathbb{R}_+} \subseteq Q$ and for all $z\in S$, $\widetilde{z.t} \in Q$ for $t$ sufficiently large.
\end{description}

Denote the set of all $(\delta, Q)$-perturbations as $\Delta(\delta, Q)$.  This set contains differential equation models that are close to the unperturbed model (\ref{eq:model}), which have solutions that eventually enter the compact set $Q$.  

\begin{definition}
(\ref{eq:model}) is robustly permanent if there is a $\delta>0$ and $\beta>0$ such that for all $(\widetilde{f}, \widetilde{g}) \in \Delta(\delta, Q)$, (\ref{permanence}) holds for all $z\in S\backslash S_0$.
\end{definition}

\begin{theorem}\label{maintheorem_pert}
The conditions in Theorem \ref{maintheorem} imply robust permanence.
\end{theorem}

To show this, we apply a result from \citet{Hirsch2001} to show permanence of $(\tilde{f},\tilde{g}) \in \Delta(\delta, Q)$ with a uniform lower bound $\beta$.  A proof is given in Appendix \ref{mainproof_pert}. It is worth noting that permanence does not in general imply robust permanence; $\dot{x}=x^2(1-x)$ is permanent but not robustly permanent. \citet{Hofbauer2004} show that robust permanence is not generic among permanent ecological equations. 

\section{Applications}\label{examples}

The main results developed here are applicable to a broad range of internal and external feedbacks. In this section, we discuss permanence in models with external environmental, internal structural and evolutionary feedbacks, which illustrate the utility of the main theorem.  In the first example, we apply our result to show how external environmental fluctuations can enable coexistence amongst competing species in the form of robust permanence.  In the second example, we demonstrate how existing permanence conditions from \citet{Hofbauer2010} for models with internal population structure, i.e., the partitioning of a whole population into distinct types, can be reproduced using our framework. Then we give an example of a sexually-structured population model to which the existing result from \citet{Hofbauer2010} does not apply, emphasizing the utility of our result to structured models.
Finally, in the third example, we apply the result to an example of an ecological model with the evolution of a quantitative trait as the internal feedback, demonstrating how our results apply to models of eco-evolutionary dynamics.
Altogether, these applications highlight how Theorem \ref{maintheorem} unifies some existing permanence results and how it enables us to determine when there is permanence in population models with a variety of feedbacks.  

\subsection{Environmental fluctuations}
Population dynamics are often influenced by time-varying environmental factors, such as seasonal fluctuations in temperature and rain fall or other weather patterns.  When environmental factors influence populations' growth rate, this may affect persistence of the community. Non-autonomous differential equations, with time-varying parameters, are commonly used to account for the temporal changes in growth rates  (e.g. \citet{Zhao2001, Vance1989, Smith2011}). These give the differential equation

\begin{equation}\label{nonauto}
\frac{dx_i}{dt}=x_if_i(x, t) \hspace{7mm} i=1\dots n
\end{equation}

where the per-capita growth rates depend on time.  

Non-autonomous models can be formulated into our model form (\ref{eq:model}) when the environmental factors can be modeled as a solution of an autonomous differential equation $\frac{dy}{dt}=g(y)$. Then (\ref{nonauto}) becomes 

\begin{equation}\label{nonauto_split}
\begin{aligned}
\frac{dx_i}{dt}&=x_if_i(x, y) \hspace{10mm} i=1\dots n \\ 
\frac{dy_j}{dt}&=g_j(y) \hspace{17mm} j=1\dots m
\end{aligned}
\end{equation}

To apply our main theorem, $y$ must remain in a compact set $K\subset \mathbb{R}^m$. Biologically, there is no mutual feedback between $y$ and $x$, which is appropriate when $y$ represents environmental factors, such as weather, that are independent of the population densities. Model (\ref{nonauto_split}) is a special case of a \emph{skew product flow}, which are commonly used for studying non-autonomous flows \citep{Zhao2001, Mierczynski2004}. 
 
To illustrate how our results can be applied to non-autonomous systems, we first prove a general, algebraically verifiable condition for non-autonomous Lotka-Volterra systems where only the ``intrinsic'' per-capita growth rates fluctuate.  Indeed, for these Lotka-Volterra systems permanence conditions are equivalent to an autonomous Lotka-Volterra system with the fluctuating intrinsic rate of growth replaced by an averaged intrinsic rate of growth. When the interaction coefficients fluctuate, however, this simplification is no longer possible. We illustrate verifying our permanence condition in this latter case for a Lotka-Volterra systems with two competing species. 

For the general result, consider a non-autonomous Lotka-Volterra system of the form 
\begin{equation}\label{eq:LV}
\begin{aligned}
\frac{dx}{dt}&= x\circ (A x +b(y))\\
\frac{dy}{dt}&= g(y)
\end{aligned}
\end{equation} 
where $\circ$ denotes component-wise multiplication i.e., the Hadamard product. The matrix $A=(a_{ij})$ corresponds to the matrix of per-capita species interaction strengths and the vector $b(y)$ corresponds to the intrinsic per-capita growth rates as a function of the ``environmental'' state $y$. As $y$ doesn't depend on $x$, we write $y.t$ as the solution of $\frac{dy}{dt}=g(y)$ with initial condition $y \in K$. 

For simplicity, we assume the dynamics of $y$ on $K$ are uniquely ergodic, i.e., there exists a Borel probability measure $\mu$ on $K$ such that 

\[ 
\overline {h} :=\lim_{t\to\infty} \frac{1}{t}\int_0^t h(y.s)ds =\int h(y)\mu(dy)
\] 
for all $\mu$-integrable functions $h:K\to\mathbb{R}$ satisfying $\int |h(x)| \mu(dx)<\infty$.
In particular, let $\overline{b}=(\overline{b_1},\dots, \overline{b_n})$ be the temporal averages of the intrinsic rates of growth. Using these averages, we prove the following two results.

\begin{proposition}\label{LVprop} Assume that \eqref{eq:LV} satisfies assumption \textbf{S2}. If there exist $p_1,\dots,p_n>0$ such that 
\begin{equation}\label{LVpropineq}
\sum_i p_i \left(\sum_j a_{ij} x_j + \overline {b_i}\right)>0
\end{equation}
for any $x\in\mathbb{R}_+^n$ satisfying $\Pi_i x_i=0$ and $\sum_j a_{ij} x_j =-\overline{b_i}$ whenever $x_i>0$, then \eqref{eq:LV} is robustly permanent. 

\end{proposition}

The proof of this proposition is in Appendix \ref{prop_proof}.

\begin{proposition}\label{LVprop2}
If there is no $x$ such that $\sum_j a_{ij} \overline{x}_j = -\overline{b_i}$ with $x_i>0$ for all $i$, then $ \omega(z)\subset S_0$ for all $z \in S\backslash S_0$. 
\end{proposition}

\begin{proof}
Following the proof of Theorem 5.1.2 in \citet{Hofbauer1998}, there exists a $p$ such that $\sum_i p_i(\sum_j a_{ij}x_j+ \overline{b_i})>0$ for all $x\in \mathbb{R}_+^n$.  Let $V(z)=\sum_i p_i \log(x_i)$ for all $z=(x,y)\in S\backslash S_0$.   Then $\frac{dV}{dt}=\sum_i p_i(\sum_j a_{ij}x_j(t)+ b_i)$.  Now, suppose there is a $z \in S\backslash S_0$ with $\omega(z)\subset S\backslash S_0$.  Then, by compactness, there is a $z^* \in \omega(z)$ such that $V(z)$ is maximized on $\omega(z)$. Also, by compactness, for all $z \in \omega(z)$,$ \overline{\frac{dV}{dt}}=\lim_{s\rightarrow\infty} \frac{1}{s} \int_0^s \frac{dV}{dt} ds >\beta$, for some $\beta>0$ but this contradicts the existence of a maximum.  It follows that for all $z\in S\backslash S_0$, $\omega(z)\not \subset S\backslash S_0$ and $\omega(z)\cap S_0 \neq \emptyset$. Then, by the Zubov-Ura-Kimura theorem~\citep{Garay2003}, $\omega(z)\subset S_0$.
\end{proof}

Hence, when environmental variation drives fluctuations in intrinsic growth rates, their effects can be averaged in time to determine permanence. On the contrary, if interaction coefficients fluctuate, then permanence may hold, even if predictions from averaging these coefficients in time suggests otherwise. 

To demonstrate this explicitly, we consider with a modified version of the autonomous model from \citet{Volterra1928} of two species competing for a single limiting resource.  
Let $x_1$ and $x_2$ be the densities of two species competing for a limited resource, $R$.  Suppose the death rate and resource use of species $i$ depend on a changing environmental state $y$ so that the intrinsic death rate $d_i(y)$ and the interaction coefficients $a_i(y)$ are functions of $y$. 
The model from \citet{Volterra1928} becomes

\begin{equation}\label{comp2}
\begin{aligned}
\frac{dx_1}{dt}&= x_1(ca_1(y)R - d_1(y))\\
\frac{dx_2}{dt}&= x_2(ca_2(y)R - d_2(y))\\
\frac{dy}{dt}&= g(y)\\
R&=\max\{J- a_1(y)x_1 - a_2(y)x_2, 0\}\\
\end{aligned}
\end{equation}
where $c$ is the efficiency with which both species convert the resource into new individuals and $J$ is the maximum amount of resource available and this is instantly reduced by the competitors. 
We assume that the dynamics of $y$ are uniquely ergodic on a compact set $K$. This model is appropriate for species in which resource use or death rate change with the seasons or a fluctuating environment.  

In the constant environment model ($g(y)=0$), \citet{Volterra1928} showed that if $\frac{d_1(y)}{a_1(y)} < \frac{d_2(y)}{a_2(y)}<Jc$, species $1$ will exclude species $2$: $\lim_{t\to\infty}x_2(t)=0$ for any initial condition $z=(x_1, x_2, y)$ satisfying $x_1x_2>0$.  This is commonly referred to in the ecological literature as the $R^*$ rule \citep{Tilman1980} and is a mathematical formulation of the competitive exclusion principle, which asserts that two competing species for the same resource cannot coexist, if other ecological factors are constant \citep{Gause1934, Hardin1960}.

Environmental fluctuations that lead to time-varying parameters might affect the coexistence of two species competing for the same resource.  Proposition \ref{LVprop2} implies that if only the per-capita death rates vary, then the competitive exclusion principle still holds.  However, when the resource use rates vary coexistence is possible.
Specifically, suppose that species $i$ uses the resource at a maximal rate for some compact subset of environmental states $K_i\subset K$ so that $a_i(y)=1$ and $a_j(y)=0$ for $y\in K_i, i\neq j$.  To allow for temporal partitioning of resource use, we assume that these sets of environmental states are disjoint i.e. $K_1 \cap K_2=\emptyset$. Let $k_i=\lim_{t\rightarrow\infty}\frac{1}{t} \int_0^t \mathbbm{1}_{K_i}(y.s)ds$ be the average time spent in environmental state $K_i$, where $\mathbbm{1}_{K_i}: K\rightarrow \mathbb{R}$ is the indicator function with $\mathbbm{1}_{K_i}(y)=1$ for $y\in K_i$ and $0$ otherwise. Furthermore, assume that $d_i(y)>\epsilon$ for some $\epsilon>0$ and for all $y$ and $i=1,2$. For example, this might model the dynamics of winter annual plants in the Sonoran desert that use water following winter rains, while summer annuals tend to do so during summer \citep{Smith1997}. 

\begin{theorem}\label{th_case2}
If $cJk_i>\overline{d_i}$ for $i=1,2$, then (\ref{comp2}) is robustly permanent.
\end{theorem}

\begin{proof}
First, note that (\ref{comp2}) satisfies \textbf{S2} with $Q=\{[0,\frac{cJ^2}{\epsilon}] \times K\}$, as $\frac{dx_i}{dt}<0$ whenever $x_i>\frac{cJ^2}{\epsilon}.$
 
Next, we show that each species persists on its own when the other species is absent. Consider the single species $i$ model 

\begin{equation}\label{comp_sub}
\begin{aligned}
\frac{dx_i}{dt}&= x_i(ca_i(y)(\max\{J-a_i(y)x_i,0\}) - d_i(y))\\
\frac{dy}{dt}&=g(y)\\
\end{aligned}
\end{equation}

on $S^i=\mathbb{R}_+ \times K$ with extinction set $S^i_0=\{0\}\times K$.  $\mathcal{M}=\{S^i_0\}$ is a Morse decomposition of $\Gamma_i \cap S^i_0$, where $\Gamma_i$ is the global attractor for (\ref{comp_sub}). Then, for all $z \in S^i_0$, 

$$
\lim_{t\rightarrow \infty} \frac{1}{t} \int_0^t f_i(0,y.s)ds=\lim_{t\rightarrow \infty} \frac{1}{t} \int_0^t (ca_i(y.s)J-d_i(y.s)) ds > cJk_i  -\bar{d_i} > 0
$$

By Theorem \ref{maintheorem}, (\ref{comp_sub}) is permanent.  Let $A_i\subset S^i\backslash S_0^i$ be the attractor in $\Gamma_i$.

Now, consider (\ref{comp2}).  Let $M_3=\{0\} \times \{0\} \times \{K\}$ and $M_i=A_i$ for $i=1,2$.  Then, $\mathcal{M}=\{M_3, M_2, M_1\}$ is a Morse decomposition of $S_0 \cap \Gamma$, where $\Gamma$ is a global attractor for (\ref{comp2}). With $p=(1,1)$, the inequality in Theorem \ref{maintheorem} is satisfied for Morse set $M_3$.  For $i=1, 2$,  

\[
\lim_{t\rightarrow \infty} \frac{1}{t} \int_0^t f_i(z.s)ds =0 
\]
and
\[
\lim_{t\rightarrow \infty} \frac{1}{t} \int_0^t f_j(z.s)ds > cJk_1 - \overline{d_1}>0
\]
 
for $j\neq i,$ for all $z\in M_i$.  Then $p$ satisfies the inequality in Theorem \ref{maintheorem} for $M_i$.  Finally, by Theorem \ref{maintheorem_pert}, (\ref{comp2}) is robustly permanent. 
\end{proof}

This result implies that even if species $1$ is on average a stronger resource competitor, i.e., $\frac{\bar{d_1}}{c\bar{a}_1} < \frac{\bar{d_2}}{c\bar{a}_2}$, it may not always exclude species $2$.  Temporal differences in resource use enable weaker competitors to coexist with stronger competitors.  The condition in Theorem \ref{th_case2} suggests that when per-capita death rates are high, the species needs a longer time period to maximally acquire the resource to ensure permanence.  Furthermore, the more resource that is available (greater $J$), the shorter this time period can be, all else being equal. This is an example of the storage effect mechanism of coexistence: species have different environmental time periods that are good for growth and are able to survive through time periods bad for growth~\citep{Chesson1981, Chesson1994}.

\subsection{Structured populations}
Individual variation that gives rise to intraspecific differences in demographic rates and species interactions can alter community dynamics and hence, persistence \citep{Moll2008, Bolnick2011, Fujiwara2011, Leeuwen2014}.  One form of structured population models account for this individual variation by partitioning  populations into discrete types, e.g. size classes, spatial location, and gender. 
For example, \citet{Hofbauer2010} considered models of interacting, structured populations of the form 
\begin{equation}\label{eq:structured}
\frac{du_i}{dt} = A^i(u) u_i
\end{equation}
where $u_i=(u_{i1},\dots,u_{im_i})$ is a vector of densities for the $m_i\geq 1$ subpopulations of species $i$, $u=(u_1,u_2,\dots,u_n)$ is the state of the entire community, and $A^i(u)=(a_{jk}^i(u))_{j,k}$ are $n_i\times n_i$ matrices with non-negative off-diagonal entries and the sign structure of an irreducible matrix that only depends on $i$. First, we show how our result reproduces a previous result from \citet{Hofbauer2010} for permanence in structured population models. Second, through a sexually-structured model, we illustrate how our result applies to models that prior results do not.

\subsubsection{Reproduce results from \citet{Hofbauer2010}}
Assume that the semi-flow defined by equation~\eqref{eq:structured}, with solutions $u.t$ for initial condition $u$, has a global attractor $\Gamma$. 
To characterize robust permanence of these equations, \citet{Hofbauer2010} used dominant Lyapunov exponents that characterize the long-term growth rates of each of the species. To define the exponents for species $i$, consider the linear skew product flow on
$\Gamma \times \mathbb{R}^{m_i}$ defined by $(u.t, v.t)=(u.t ,B_i(t,u)v)$ where
$Y(t)=B_i(t,u)$ is the solution to $Y' (t)= A_i(u.t) Y(t)$ with
$Y(0)$ equal to the identity matrix. 
The assumption that $A_i$ is irreducible with non-negative off diagonal entries implies that $B_i(t,u)\mathbb{R}_+^{m_i}\subset \mathbb{R}_+^{m_i}$ for all $x$ and $t>0$ (see, e.g., \cite{Smith1995}). \citet[Prop.3.2]{Ruelle1979} provides a non-autonomous form of the Perron-Frobenius Theorem: there exist continuous maps $v_i,w_i:\Gamma \to \mathbb{R}_+^{m_i}$ with
$\|v_i(u)\|=\|w_i(u)\|=1$, where $\|v\|=\sum_i|v_i|$, such that
\begin{itemize}
\item The line bundle $E_i(u)$ spanned by $v_i(u)$ is invariant, i.e.,
$E_i(u.t)=B_i(u,t)E_i(u)$ for all $t\geq 0$.
\item The vector bundle $F_i(u)$ perpendicular to $w_i(u)$ is invariant
i.e., $F_i(u.t)=B_i(u,t)F_i(u)$ for all $t\geq 0$.
\item There exist constants $\alpha>0$ and $\beta>0$ such that
\begin{equation}\label{perron}
\|B_i(t,u)|F_i(u)\| \leq \alpha \exp(-\beta t) \|B_i(t,u)|E_i(u)\|
\end{equation}
for all $u\in \Gamma$ and $t\geq 0$.
\end{itemize}
In light of \eqref{perron}, $v_i(u)$ can be viewed as the community state-dependent ``stable stage distribution'' of species $i$ for the linearized dynamics given by $Y'(t)=A_i(u.t)Y(t)$. Specifically, \eqref{perron} implies that for any  $\tilde v\in\mathbb{R}_+^{m_i}\setminus\{0\}$, $Y(t)\tilde v/\|Y(t) \tilde v\|-v(u.t)$ converges to zero at $t\to\infty$. Similarly, $w_i(u)$ can be interpreted as the community state-dependent vector of ``reproductive values'' for the stages of species $i$.  Stages with larger entries in $w_i(u)$ contribute more to the long-term growth rate of species $i$. 

\citet{Hofbauer2010} defined the average per-capita growth rate of species $i$ given the initial community state $u$ as 
\[
r_i(u)=\limsup_{t\to\infty} \frac{1}{t}\int_0^t  w_i(u.s)^T A_i(u.s) v_i(u.s)\,ds
\] 
where $w^T$ denotes the transpose of a vector $w$. We derive the following theorem of \citet{Hofbauer2010} as a corollary of Theorem~\ref{maintheorem}. 

\begin{theorem}\label{thm:hofsch}
Let $\{M_1,\dots, M_\ell\}$ be a Morse decomposition for $S_0\cap \Gamma$. If for each $M_k$ there exists $p_{k1},\dots,p_{kn}>0$ such that 
\begin{equation}\label{eq:hofsch}
\sum_i p_{ki} r_i (u)>0
\end{equation}
for all $u\in M_k$, then system~\eqref{eq:structured} is robustly permanent.\end{theorem}

\begin{proof}
To prove Theorem~\ref{thm:hofsch} using our framework, we introduce the following change of variables:
\[
x_i = \sum_j u_{ij} \mbox{ and } y_{ij}=u_{ij}/x_i.
\]
In this coordinate system, equation~\eqref{eq:structured} becomes
\begin{eqnarray}\label{eq:structured2}
\nonumber\frac{dx_i}{dt}&=& x_i \sum_{j,k} b_{jk}^i(x,y) y_{ik} =: x_i f_i(x,y) \mbox{ where }b_{jk}^i(x,y)= a_{jk}^i(u)\\
\frac{dy_{ij}}{dt}&=& \left( \sum_k b_{jk}^i(x,y) y_{ik} -  y_{ij}f_i(x,y)\right)=:g_{ij}(x,y).
\end{eqnarray}
The state space for equation \eqref{eq:structured2} is $\tilde S=\mathbb{R}^n_+\times \Delta_{m_1} \times \dots \Delta_{m_n}$ where $\Delta_k = \{ y\in \mathbb{R}^k_+ : \sum_j y_j =1\}$ is the $k-1$ dimensional simplex. Let $\tilde \Gamma \subset \tilde S$ and $\{\tilde M_k \}_{k =1}^\ell$ be the global attractor $\Gamma$ and the Morse decomposition $\{M_k \}_{k =1}^\ell$, respectively, of equation~\eqref{eq:structured} in this coordinate system. 

Fix  an element $\tilde M_k$ of the Morse decomposition and $z=(x,y)\in \tilde M_k$.  Let $u$ be $z$ in the original coordinate system. Proposition 1 of \citet{Hofbauer2010} implies that 
\[
r_i(u)=\liminf_{t\to\infty}\frac{1}{t}\int_0^t f_i(z.s)ds.
\]
By the assumption of the theorem statement, 
\[
\sum_i p_i r_i(u)>0.
\]
Hence, we can choose $T_z>0$ such that 
\[
\sum_i p_i \int_0^{T_z} f_i(z.s)ds >0.
\]
Applying Theorem~\ref{maintheorem} completes the proof. 
\end{proof}

The change of variables from (\ref{eq:structured}) to (\ref{eq:structured2}) demonstrates how structured populations can be reformulated into our general framework and reproduce results from \citet{Hofbauer2010}.

\subsubsection{Sexually structured populations}
Our main permanence result applies to structured models that \citet{Hofbauer2010} does not. In particular, permanence results from \citet{Hofbauer2010} do not apply to models in which growth depends on the frequency of types in the populations. 

As an example, we consider a rock-paper-scissors three-species competition model, in which each species is sexually-structured such that reproduction depends on the frequencies of males and females.  Let $m_i$ be the density of males and $f_i$ the density of females for species $i$. Following \citet{Caswell1986}, we assume that there is a harmonic mating function in which case the rate at which females and males are produced (assuming a 50-50 primary sex-ratio) is 
\[
b \frac{m_if_i}{f_i+m_i}
\]
where $2b$ is the per-capita birth rate of mated females, which is species-independent. Assume also that mortality is species-independent but sex-specific, with $d_m$ and $d_f$ as the per-capita, density-independent mortality rates of males and females, respectively. To account for intra- and inter-specific density-dependent feedbacks due to competition,  let $a_{ij}$ be the strength of the competitive effect of species $j$ on species $i$. For simplicity, we assume these density-dependent effects are not sex-specific. However, the model can be easily modified to account for these sex-specific feedbacks. Under these assumptions, the model is

\begin{equation}\label{sex_model}
\begin{aligned}
\frac{df_i}{dt}&=f_i \left( b  \frac{m_i }{f_i+m_i} - d_f - \sum_j a_{ij}(m_j+f_j)\right)\\
\frac{dm_i}{dt}&=m_i \left( b  \frac{f_i }{f_i+m_i} - d_m - \sum_j a_{ij}(m_j+f_j)\right)\\
i&=1,2,3
\end{aligned}
\end{equation}

To ensure each species can persist in the absence of the others, we assume that $b>d_m+d_f$.  To account for rock-paper-scissors competitive dynamics, we assume the interaction terms $a_{ij}$ are given by 

\[
A=a+\begin{pmatrix}
0&\beta&-\alpha\\
-\alpha& 0& \beta\\
\beta&-\alpha&0
\end{pmatrix}
\]
where $a,\alpha,\beta$ are all positive and $\alpha<a$.  

Due to the frequency dependent terms, this model does not satisfy the continuity assumptions of \citet{Hofbauer2010} and, consequently, their results can not be applied directly to study permanence of these equations. However, through the change of variables, 
\[
x_i =m_i+f_i \mbox{ and } y_i =\frac{f_i}{x_i}
\]

Equation (\ref{sex_model}) transforms to

\begin{equation}\label{sex_model2}
\begin{aligned}
\frac{dx_i}{dt}&= x_i \left( 2b y_i(1-y_i) - d_f y_i - d_m (1-y_i) -\sum_j a_{ij} x_j\right)\\
\frac{dy_i}{dt}&=y_i(1-y_i) ( b+d_m-d_f -2b y_i) \hspace{50mm} i=1,2,3
\end{aligned}
\end{equation}

and our permanence theorem applies to prove 

\begin{theorem}\label{sex_theorem}
Under these assumptions, if $\alpha>\beta$, then (\ref{sex_model2}) is permanent in $\mathbb{R}_+^3 \times (0,1)^3$. Conversely, if $\alpha<\beta$, then (\ref{sex_model2}) is not permanent.
\end{theorem}

\begin{proof}
First, note that (\ref{sex_model2}) satisfies the assumptions of our main theorem (\ref{maintheorem}).  The dynamics on the extinction set consist of an unstable equilibrium at $(x,y)=(0, 0, 0)\times (\hat y_1, \hat y_2, \hat y_3)$ and a heteroclinic cycle between single species equilibria (e.g. $(\hat x_1, 0, 0) \times (\hat y_1, \hat y_2, \hat y_3))$ where  
\[
\hat x_i = \frac{b-d_m-d_f}{a},\,\,  \hat y_i = \frac{1}{2} + \frac{d_m-d_f}{2b} \mbox{ and }  x_j= y_j=0 \mbox{ for }j\neq i.  
\]
At these equilibria, the per-capita growth rates of the missing species are $\alpha \hat x_i$ and $-\beta \hat x_i$. Using the Morse decomposition consisting of the zero equilibrium and the heteroclinic cycle,  Theorem \ref{maintheorem} implies that permanence occurs if there exist $p_i>0$ such that
\[
\begin{aligned}
p_1 \times 0 + p_2 \times \alpha \hat x_1 + p_2 \times (-\beta \hat x_1)&>0\\
p_1 \times (-\beta \hat x_2) + p_2 \times 0 + p_2 \times \alpha \hat x_2&>0\\
p_1 \times \alpha \hat x_3 + p_2 \times (-\beta \hat x_3) + p_2 \times0&>0.
\end{aligned}
\]
As $\hat x_1=\hat x_2=\hat x_3$, there is a solution to these linear inequalities if and only if $\alpha>\beta$. Conversely, there is a solution to the reversed linear inequalities if and only if $\beta>\alpha$ and then Theorem \ref{maintheorem} implies that (\ref{sex_model2}) is not permanent.
\end{proof}

Theorem \ref{sex_theorem} yields the same permanence condition as in the classic asexual model. 
Due to our assumption that density-dependent feedbacks are not sex-specific, the system is only partially coupled as the frequency dynamics of $y$ do not depend on $x$. With sex-specific density-dependent feedbacks, the system would be fully coupled but still analytically tractable as these feedbacks would appear as linear functions of $x_i$ in the $y_i$ equations.

\subsection{Quantitative genetics}
In recent years, empirical studies have demonstrated that feedbacks between evolutionary and ecological processes (eco-evolutionary feedbacks) can affect coexistence of species \citep{Lankau2007}.  As a consequence of the growing empirical evidence, theoreticians have developed models that integrate commonly used ecological models with evolutionary equations to study eco-evolutionary feedbacks.  For the evolution of quantitative traits, such as body size, a common approach is to assume that the rate of trait change is proportional to the gradient of per-capita growth with respect to the trait \citep{Lande1976}.
This has led to models of the form

\begin{equation}\label{ecoevo}
\begin{aligned}
\frac{dx_i}{dt}&=x_if_i(x, y) \hspace{3mm} i=1\dots n\\
\frac{dy}{dt}&=\sigma_G^2 \frac{\partial f_j}{\partial y}
\end{aligned}
\end{equation} 

where $y$ represents the mean of an evolving quantitative trait of one of the species $j$, and $\sigma_G^2$ is the heritable variance of the trait \citep{Lande1976}.  
These feedbacks are immediately in the form of (\ref{eq:model}) and we can use Theorem \ref{maintheorem} to identify when eco-evolutionary feedbacks mediate coexistence.

For illustrative purposes, we consider a model developed by \citet{Schreiber2011}. They consider the apparent competition community module, in which two prey species with densities $x_1, x_2$, respectively, share a common predator with density $x_3$.  In this model, the predator population has a quantitative trait that determines the attack rate of individual predators on each prey species.  The quantitative trait is assumed to be normally distributed with variance $\sigma$ in the predator population with mean $y\in [\theta_1, \theta_2]$, where $\theta_i$ is the optimal trait for attacking prey $i$.  They derived a function $a_i(y)$ of the average attack rate of the predator on prey $i$ that decreases with the distance between the trait $y$ and $\theta_i$, given by

\[
a_i(y)=\frac{\alpha_i \tau}{\sqrt{\sigma^2+\tau^2}} \exp\Bigl[-\frac{(y-\theta_i)^2}{2(\sigma^2+\tau^2)}\Bigr]\,.
\]
where $\alpha_i$ is the maximum attack rate on prey $i$ and $\tau>0$ determines how specialized the predator must be to attack each prey.
The coupled dynamics are

\begin{equation}\label{ac}
\begin{aligned}
\frac{dx_i}{dt}&= x_i (r_i(1-x_i/K_i) - x_3 a_i(y)) \hspace{6mm} i=1,2\\
\frac{dx_3}{dt}&= x_3 \,f_3(x, y)\\
\frac{d y}{dt} &= \sigma_G^2  \frac{\partial f_3}{\partial y}
\end{aligned}
\end{equation}
where $K_i>0$ and $r_i>0$ are the carrying capacity and intrinsic growth, respectively, for prey $i$. $f_3$ is the average per-capita growth rate or fitness of the evolving species given by, 
\[
f_3(x,y)=\sum_{i=1}^2 e_i a_i(y) x_i - d
\]

where $e_i>0$ is the efficiency at which the predator converts prey $i$ into new predators and $d>0$ is the intrinsic death rate of the predator.

We can apply Theorem \ref{maintheorem} to characterize permanence of this system.

\begin{theorem}~\label{thm:acpermanence} 
Let $W=\{y\in [\theta_1, \theta_2] | \frac{\partial f_3}{\partial y}(K_1, K_2, y)=0\}$ be the set of equilibria for the trait dynamics when the prey are at carrying capacity and the predator density is zero. 
If 
\begin{enumerate}
\item $\frac{r_i}{a_i(\theta_j)} > \frac{r_j}{a_j(\theta_j)} (1-\frac{d}{a_j(\theta_j)e_jK_j})$ for $i=1,2; i\neq j$ and 
\item $e_1a_1(y^*)K_1 + e_2a_2(y^*)K_2 > d$ for all $y^* \in W$
\end{enumerate}
 then the system is robustly permanent in $\mathbb{R}_+^3 \times [\theta_1, \theta_2]$.  Conversely, if either condition is not met, then the system is not permanent. 
 \end{theorem}

The first condition ensures that prey species $i$ has positive per-capita growth when the predator has evolved to optimize on prey $j\neq i$ ($y=\theta_j)$ and the predator and prey $j$ are at their unique equilibrium densities.  The second condition ensures that when the predator is rare and both prey are at carrying capacity, the predator has positive growth when it evolves to one of its trait equilibria. Using a different approach, \citet{Schreiber2015} show (\ref{ac}) is permanent under these conditions. Our results strengthen their results by showing robust permanence.

\begin{proof}
Equation (\ref{ac}) satisfies the assumptions of Theorem \ref{maintheorem}.  In particular, there is a global attractor $\Gamma$.
Let $M_6= \{(0,0,0)\} \times [\theta_1, \theta_2], M_5=\{(K_1, 0, 0, \theta_1)\}, M_4=\{(0, K_2, 0, \theta_2)\}, M_3=\{(\hat{x}_1, 0, \hat{x}_3^{(1)}, \theta_1)\}, $ and $M_2=\{(0, \hat{x}_2, \hat{x}_3^{(2)}, \theta_2)\}$ where $\hat{x}_i=\frac{d}{e_ia_i(\theta_i)}$ and $\hat{x}_3^{(i)}=\frac{r_i(1-\frac{\hat{x}_i}{K_i})}{a_i(\theta_i)}$.  
Finally, let $M_1= \{(K_1, K_2, 0)\}\times [y_1, y_2]$ where $y_1=\min_{y\in W}y$ and $y_2= \max_{y\in W}y$. 
Schreiber and Patel (2015) consider three separate cases: (i) $d\geq a_1(\theta_1)e_1K_1$, (ii) $a_1(\theta_1)e_1K_1 >d\geq  a_2(\theta_2)e_2K_2$ or (iii) $a_2(\theta_2)e_2K_2>d$.  They show that $\mathcal{M}_1= \{M_1, M_4, M_5, M_6\}$, 
$\mathcal{M}_2= \{M_1, M_3, M_4, M_5, M_6\}$ and
 $\mathcal{M}_3= \{M_1, M_2, M_3, M_4, M_5, M_6\}$ form a Morse decomposition for (\ref{ac}) under case (i), (ii), and (iii) restricted to $\Gamma \cap S_0$, respectively.

Consider case (iii).  For each Morse set  $M_k\in \mathcal{M}_3$, there exist a vector $p_k$ that satisfies the inequality in Theorem \ref{maintheorem} for every point in the set. 
For example, for $\epsilon$ sufficiently small, $\vec{p}_6=(1, 1, \epsilon)$ satisfies the inequality in Theorem \ref{maintheorem} for $M_6$.  Case (ii) and case (i) follow similarly.
\end{proof}

\section{Discussion}
Understanding how abiotic and biotic factors determine coexistence of interacting species is a fundamental problem in ecology.  Ecologists have demonstrated that factors internal to the populations, such as individual variation \citep{Bolnick2011, Violle2012, Hart2016, Barabas2016} and evolution \citep{Lankau2007, Barabas2016}, and factors external to the populations, such as temporal variation in abiotic factors \citep{Hutchinson1961}, can have substantial impacts on population dynamics. Moreover, as these internal and external factors change, their influence on population growth lead to changes in population densities which in turn may alter these factors, thereby creating a feedback loop.  The instrumental role of this feedback on coexistence has been demonstrated both empirically (e.g. \citet{Lankau2007, Chung2016}) as well as theoretically (e.g. \citet{Bever1997, Revilla2013}).  Our work develops the mathematical framework for finding conditions that enable coexistence in community dynamic models with feedbacks and, by applying this theory, elucidates the role of these internal and external feedbacks on coexistence.

We find that if there is a weighting of the species such that the temporal average of the weighted per-capita growth rates is positive whenever a species is missing, then populations coexist with feedbacks. Moreover, given a Morse decomposition of the extinction state, these weightings can differ among the components of this decomposition. For models without feedbacks, this sufficient condition for robust permanence is equivalent to the condition found by \citet{Garay2003}. Hence, our results provide a natural extension to models with internal and external feedbacks. As our examples illustrate, these feedbacks play two critical roles for coexistence.  First, the effect of the feedback variable will influence the Morse decomposition of the extinction set.  Second, feedbacks affect the per-capita growth of each of the species and thereby, influence whether the weighted combination of these growth rates can be positive.  These differences can drive feedbacks to enable or prevent coexistence.

In addition to extending the work of \citet{Garay2003} to include internal and external feedbacks, our general framework and permanence result incorporates existing population models with specific types of internal feedbacks~\citep{Hofbauer2010, Caswell1986, Caswell2001}, external feedbacks \citep{Armstrong1976, deMottoni1981, Zhao2001}, and mixtures of internal and external feedbacks \citep{Hastings2007, Cuddington2009}. Through our examples, we illustrate how to transform several of these earlier results into our framework.  In our first example, we formulate a non-autonomous model with parameters that vary with the environment into our framework by introducing a feedback variable that models the dynamics of the environmental variation. In the second example, we transform a structured population model, in which populations are partitioned into distinct types, into our framework, via a change of variables into frequencies of types and total densities.  The frequencies of the different types within the population act as the internal feedback variables. Finally, in the third example, we demonstrate how our results apply to population models with feedbacks due to  trait evolution. These examples highlight that our framework can help elucidate how populations coexist in a range of ecological scenarios.

In an attempt to explain empirical evidence for coexistence that was incompatible with theoretical predictions, \citet{Hutchinson1961} postulates that changes in the environment that alter the competitive superiority of one species over another can enable coexistence. Our first example highlights that ergodic environmental variation that drives temporal differences in ecological parameters can enable coexistence of populations in a community, but that this depends on the role it has on influencing population growth. In particular, we show that environmental variation that influences species interactions enables coexistence, in comparison to an analogous model that uses time-averaged parameters instead of explicitly accounting for variation.  Our results are an extension of previous work that showed coexistence amongst two competing species with periodic environmental variation (Armstrong and McGehee 1976, Cushing 1980, de Mottoni and Schiaffino 1981) and a specification of the general results for non-autonomous two species models (Zhao 2001).  Notably, our example demonstrates how variation and separation of resource use between two species can enable coexistence through a storage effect ~\citep{Chesson1981}, provided that species are ``stored" through periods they do not use the resource and can sufficiently recover through periods in which they do.  Interestingly, our results also highlight that environmental variation that only influences non-interaction terms, such as per-capita mortality, does not enable coexistence due to the linearity of non-interaction terms in the model.  The necessity for temporal variation to act in nonlinear ways to enable coexistence was also noted for models with stochastic environments \citep{Schreiber2010} as well as in discrete time models with non-overlapping generations \citep{Chesson1981, Chesson1994}. 

In addition to externally-driven temporal variation, internal variation within populations may also impact coexistence. Many reviews highlight that models with internal variation can lead to different predictions and inferences in both empirical and theoretical ecological studies compared to mean field models \citep{Bolnick2011, Violle2012, Hart2016}. Structured population models, a commonly used framework for accounting for internal variation, involve partitioning the population into distinct types, such as based on sex, life stages, or location in space, so that each type has its own growth rate depending on all other types \citep{Caswell1986, Caswell2001}.  Our permanence condition can be used to determine when structured interacting populations coexist. These results apply to structured models that previous results from \citet{Hofbauer2010} do not. Mainly, \citet{Hofbauer2010} made two mathematical assumptions.  First, they assume that there were no negative interactions between individuals of different types.  This assumption may not hold in a number of common ecological scenarios, including models with cannibalism or other forms of interference, which is a prevalent negative interaction between different life stages within a population \citep{fox1975cannibalism, polis1981cannibalism}. Second, they assume continuity in the growth matrices $A_i$, which restricts their framework to models with no frequency dependent growth.  Growth in structured models may be frequency dependent in a number of biological scenarios.  In our example, we apply our results to a sexually-structured model, which, following \citet{Caswell1986}, has frequency dependence since fecundity depends on sex ratios.  In particular, through a change of variable from the densities of different types within a species to total density and frequency of types, structured models can be reformulated into our framework, making the permanence conditions applicable to a broad range of structured models.
 
When individual variation in a population is heritable, this sets the stage for evolution to take place in response to differential selection pressures \citep{Violle2012}. Recent empirical evidence has demonstrated the prevalence of feedbacks between population dynamics and trait evolution (called eco-evolutionary feedbacks; reviewed in \citet{Schoener2011, Lankau2011c} and that these feedbacks may impact population dynamics \citep{Abrams2000, Cortez2010, Vasseur2011, Lankau2009, Schreiber2011, Northfield2013, Patel2015}.  Thus far, few studies have shown permanence in these types of models (but see \citep{Schreiber2011, Schreiber2015}), and we hope that these results will motivate analyses of coexistence in the sense of permanence in future eco-evolutionary studies.  Through our example, we demonstrate how these results can elucidate the conditions for robust permanence in a model where a predator is evolving between traits that are more fit for attacking one prey species versus another. In the absence of eco-evolutionary feedbacks, the prey species exhibit apparent competition: increasing the density of one prey increases the predator density and, thereby, results in a reduction of the other prey species~\citep{Holt1977}.  For highly enriched environments in which the carrying capacities of the prey are large, this apparent competition can lead to exclusion of one of the prey species~\citep{Holt1994}. As the predator evolves to specialize on the more common prey, eco-evolutionary feedbacks can rescue the rare prey from this outcome and enable coexistence. Applying our results to other eco-evolutionary models may provide opportunities to gain a more general understanding of the role of evolution on species coexistence.

Our results here extend existing methods for permanence to account for internal and external feedbacks, generalizing some existing results and broadening their applicability. There are a number of natural avenues that would be useful to develop in the future, including infinite dimensional models and stochastic models.  We assume feedbacks are contained in $\mathbb{R}_+^m$.  However, some internal and external feedbacks may be better captured in infinite dimensions and extending our results to account for this may be useful (e.g. integral projection models; \citet{Easterling2000}).  Permanence has been studied in general infinite dimensional dynamical systems \citep{Hale1989} as well as in models with specific types of feedbacks, including those captured through continuous spatial heterogeneity  \citep{dunbar1986, cantrell1993, cantrell1996, Cantrell2003, zhao1994, furter1997, Mierczynski2004} and time delays \citep{Burton1989, Freedman1995, Ruan1999}.  Whether transforming these models into a framework analogous to the one here is useful requires further exploration. Furthermore, populations may feedback with random internal or external factors. Extending permanence results for stochastic population models to account for feedbacks will enable comparisons to our framework to understand broadly the role of random feedbacks on coexistence.  
With the growing number of empirical studies investigating internal and external factors that influence population dynamics, ecological models are becoming more sophisticated.  In order for permanence to remain an important concept in ecology, the methods for demonstrating permanence must continue to expand to these new modeling frameworks.  

\section*{Acknowledgements}
Financial support by the U.S. National Science Foundation Grant DMS-1313418 to SJS, the American Association of University Women Dissertation Fellowship to SP and the Austrian Science Fund (FWF) Grant P25188-N25 to Reinhard Burger at the University of Vienna is gratefully acknowledged.

%%%%%%%%%%%%%%%%%%%%%             APPENDIX             %%%%%%%%%%%%%%%%%%%%%%%

\setcounter{section}{0}
\renewcommand{\thesection}{A\arabic{section}}
\renewcommand{\thefigure}{A\arabic{figure}}
\setcounter{figure}{0}
\setcounter{equation}{0}
\renewcommand\theequation{A\arabic{equation}}
\renewcommand\thedefinition{A\arabic{definition}}

\section{Proof of Theorem \ref{maintheorem}}\label{mainproof}
To prove Theorem \ref{maintheorem}, we begin by introducing ``good average Lyapunov functions" and proving a more general theorem. 

\begin{definition}
A continuous map $P: U \rightarrow \mathbb{R}$, where $U\subset S$ is an open set, is a good average Lyapunov function (GALF) for (\ref{eq:model}) if 
\begin{itemize}
\item $P(z)=0$ for all $z\in S_0 \cap U$ and $P(z)>0$ for all $z\in$ $(S\backslash S_0) \cap U$,
\item $P$ is differentiable on $(S\backslash S_0) \cap U$,
\item $\frac{\partial P}{\partial y_j}=0$ for all $j$,
\item  $p_i:=\frac{x_i}{P}\frac{\partial P}{\partial x_i}$, which are continuous functions defined on $(S\backslash S_0)\cap U$ and extend continuously to $S \cap U$, and

\item for every $z\in  S_0 \cap U$, there is a $T_z>0$ such that $z.t\in U$ for $t\in [0,T_z]$, and 

$$\int_0^{T_z} \sum_i p_i(z.t)f_i(z.t) dt >0.$$

\end{itemize}
\end{definition}

We prove the following theorem

\begin{theorem}
Let $\mathcal{M}=\{M_1, M_2, \dots M_\ell\}$ be a Morse Decomposition for (\ref{eq:model}) restricted to $S_0\cap \Gamma$.  
For each $k$, let $U_k$ be an open neighborhood of $M_k$ and let $P_k: U_k \rightarrow \mathbb{R}$ be a good average Lyapunov function for (\ref{eq:model}).  Then (\ref{eq:model}) is permanent.
\end{theorem}

\begin{proof}
Fix $k$.  Let $F(z, T):=\int_0^T \sum_i p_i(z.t)f_i(z.t) dt$ for all $z\in U_k$ and $T\geq 0$, where $p_i$ is defined from the definition of a GALF. 

By definition of the GALF, for all $z \in M_k$, there is a $T_z>0$ and $\delta(z)>0$ such that
 \[
F(z,T_z) > \delta(z)
\]

By continuity of $F$, there is a neighborhood $V_z\subset U_k$ of $z\in M_k$ such that $F(v, T_z)>\frac{\delta(z)}{2}$ for all $v\in V_z$. The collection of sets $\{V_z\}_{z\in M_k}$ forms an open cover of $M_k$. By compactness, there is a finite subcover $\{V_{z_j}\}_{j=1}^\ell$.  Let $V_k=\cup_{j=1}^\ell V_{z_j}, c=\frac{1}{2}\min\{\delta(z_j)\}_{j=1}^\ell$ and $T=\max\{T_j\}_{j=1}^\ell$, where $T_j:=T_{z_j}$. 

Then, $V_k\subset U_k$ is a neighborhood of $M_k$ such that for all $z\in V_k$ there is a $0<T_j<T$ satisfying
\[
F(z,T_j) >c.
\]

Furthermore, for $z\in (S\backslash S_0)\cap V_k$,
\begin{equation*}
\begin{aligned}
\ln(P(z.T_j))- \ln(P(z)) &= \int_0^{T_j} \frac{1}{P(z.s)} \frac{d}{ds} P(z.s)ds\\
& = \int_0^{T_j} \frac{1}{P(z.s)} (\sum_{i=1}^n \frac{\partial P}{\partial x_i}\frac{dx_i}{ds} + \sum_{i=1}^m  \frac{\partial P}{\partial y_i}\frac{dy_i}{ds}) ds\\
& = \int_0^{T_j} \sum_{i=1}^n p_i(z.s)f_i(z.s) ds > c,
\end{aligned}
\end{equation*}
which gives

\begin{equation}\label{ineq}
P(z(T_j))> (1+c)P(z).
\end{equation}

By the Corollary to Theorem 2 from \citet{Garay1989}, permanence follows from showing that each $M_k$ is isolated and that $(S\backslash S_0) \cap W^s(M_k)= \emptyset$, where $W^s(M_k)=\{z\in S| \emptyset\neq \omega(z)\subset M_k\}$. For any initial condition $z$, let $\gamma^+(z)=z.[0,\infty)$ be the forward trajectory of $z$.  
Assume there is a $z\in (S\backslash S_0)\cap V_k$ such that $\gamma^+(z)\subseteq V_k$. Then, there is a $z^* \in \overline{\gamma^+(z)}\cap V_k$ such that $P(z^*)=\max_{v\in\overline{\gamma^+(z)}\cap V_k} P(v)$.  Then, either (i) there exists a $t^*>0$ such that $z^*=z.t^*$ or (ii) there exists $t_n \rightarrow \infty$ such that $z_n:=z.t_n$ converges to $z^*$ as $n\rightarrow \infty$.  If (i), then equation (\ref{ineq}) implies $P(z^*.T_{z^*})>(1+c)P(z^*)$ for some $T_{z^*}>0$, which is a contradiction to the choice of $z^*$ since $z^*.T_{z^*} \in \overline{\gamma^+(v)}$.  If (ii), then for some sequence $T_{z_n}>0, P(z_n.T_{z_n})> (1+c) P(z_n) \rightarrow (1+c) P(z^*)$, which is a contradiction to the choice of $z^*$ since $z_n.T_{z_n}\in \overline{\gamma^+(z)}$ for all $n$.

Hence, for all $v\in (S\backslash S_0)\cap V_k$, $\gamma^+(v)\backslash V_k \neq \emptyset$. 
It follows that $M_k$ is isolated and for all $v \in (S\backslash S_0)\cap V_k$, $\omega(v) \not\subset M_k$.  The latter gives that $(S\backslash S_0)\cap W^s(M_k)  = \emptyset$.

\end{proof}

Theorem \ref{maintheorem} immediately follows when using the standard form (\ref{GALFstandard}) as a GALF on each Morse set.  It is easy to see that when $p_{ki}>0$ for all $i,k$, (\ref{GALFstandard}) satisfies the first four properties of a GALF. In particular, the fourth and fifth property follow since $\frac{x_i}{P}\frac{\partial P}{\partial x_i}=p_{ki}$ and hence is constant in $U$.

\section{Proof of Theorem \ref{maintheorem_pert}}\label{mainproof_pert}
To show robust permanence, will use the following theorem from \citet[Corollary 4.5]{Hirsch2001}. Note that Corollary 4.5 \cite{Hirsch2001} is for maps, but an analogous proof for flows holds (see \citet{Hirsch2001} Remark 4.6).

\begin{theorem}\label{Hirsch}
Let $(S, d)$ and $(\Lambda, \rho)$ be metric spaces.  For each $\lambda\in \Lambda$, let $\phi_\lambda : S\times \mathbb{R} \rightarrow S$ be a flow that is continuous in $\lambda, z, t$. Let $S_p\subset S$ be an open subset such that $S_p$ is invariant for all $\lambda$ and let $\partial S=S\backslash S_p$. Also, assume that every forward trajectory for $\phi_\lambda$ has compact closure in $S$ and that $\bigcup_{\lambda \in \Lambda, z\in S} \omega_\lambda (z)$ has compact closure, where $\omega_\lambda$ denotes the $\omega$-limit for $\phi_\lambda$. Let $\lambda_0\in\Lambda$ be fixed and assume that
\begin{description}
\item[T1] $\phi_{\lambda_0}$ has a global attractor $\Gamma$ and there exists a Morse Decompositions $\{M_1, \dots, M_\ell\}$ on $\Gamma\cap \partial S$
\item[T2] there exists $\delta_0>0$ such that for any $\lambda\in\Lambda$ with $\rho(\lambda, \lambda_0)<\delta_0$ and any $z\in S_p$, \\
$\limsup_{t\rightarrow\infty} d(\phi_\lambda(z,t), M_k)\geq \delta_0,$ for all $1\leq k\leq\ell.$
\end{description}

Then, there exists a $\beta>0$ and $\delta>0$ such that $\liminf_{t\rightarrow\infty} d(\phi_\lambda(z,t), \partial S)\geq \beta$ for any $\lambda\in\Lambda$ with $\rho(\lambda,\lambda_0)<\delta$ and any $z\in S$.
\end{theorem}

To apply this theorem, endow $\Lambda=\Delta(1,Q)$ with the sup norm (i.e. $\rho(h_1,h_2)=\|(h_1,h_2)\|_\infty = \sup_{z\in Q} \| h_1(z),h_2(z)\|$). Let $S=\mathbb{R}^{n+m}$ with the standard metric and let $S_p=S\backslash S_0$. By \textbf{R3} in the definition of perturbations, $\bigcup_{(\widetilde{f},\widetilde{g}) \in \Delta(1,Q), z\in S} \widetilde\omega(z)\subset Q$ and so has compact closure.  For $(f,g)$, we have a Morse Decomposition on $\Gamma \cap S_0$. 

Hence, we only have to show \textbf{T2}.  To do this, we show that there exists a $1>\delta>0$ sufficiently small such that if $P_k$ is a GALF on $U_k$ for (\ref{eq:model}), then it is also a GALF on $V_k$ for every $(\tilde{f}, \tilde{g}) \in \Delta(\delta, Q)$, where $V_k$ is defined in the proof of Theorem \ref{maintheorem}. 
First, note that the first four conditions defining a GALF are properties of $P$ and independent of the flow.  Hence, these are still satisfied.  We must show the final condition: for a $\delta>0$ sufficiently small, for all $z\in V_k \cap S_0$, there is a $T_z>0$ such that

\begin{equation}
\int_0^{T_{z}} \sum_{i=1}^n p_i(\tilde{z}.t)\tilde{f}_i(\tilde{z}.t) dt > 0
\end{equation}

for every $(\tilde{f}, \tilde{g}) \in \Delta(\delta, Q)$.

Let $F(z, T):=\int_0^T \sum_i p_i(z.t)f_i(z.t) dt$ and $\tilde{F}(z, T):=\int_0^T \sum_i p_i(\tilde{z}.t)\tilde{f}_i(\tilde{z}.t) dt$.

Let $T,c$ be such that for all $z \in V_k$, there is a $0<T_z<T$ for which $F(z, T_z)>c$, as in the proof of Theorem \ref{maintheorem}.  Then, for sufficiently small $\delta>0$,

\begin{equation*}
\begin{aligned}
|\tilde{F}(z, T_z)-F(z, T_z)|&\leq \int_0^{T_z} \sum_1^n |p_i(\tilde{z}.t)\tilde{f}_i(\tilde{z}.t)- p_i(\tilde{z}.t)f_i(\tilde{z}.t)| + |p_i(\tilde{z}.t)f_i(\tilde{z}.t)- p_i(z.t)f_i(z.t)|dt\\
&<\frac{c}{2}
\end{aligned}
\end{equation*}

for every $(\tilde{f}, \tilde{g}) \in \Delta(\delta, Q)$ and all $T_z\in [0,T]$.  The first inequality follows from triangle inequality.  The second inequality follows from \textbf{R1} in the definition of perturbations, which constrains the first difference in the sum, and Gronwall's inequality and Lipschitz continuity of $x_if_i$ and $g_i$, which constrain the second difference in the sum.

Finally, for all $z\in  S_0 \cap V_k$, there is a $0<T_z<T$ such that

$$\tilde{F}(z, T_z)\geq -|\tilde{F}(z, T_z)-F(z, T_z)| + F(z, T_z)\geq -\frac{c}{2} +c=\frac{c}{2}>0$$

Hence, for sufficiently small $\delta>0$, $P_k$ is a GALF on isolating neighborhood $V_k$ for every $(\tilde{f}, \tilde{g}) \in \Delta(\delta, Q)$. By the same argument as in the proof of Theorem \ref{maintheorem}, this implies that for all $z\in S_p \cap V_k$, $\gamma^+(z)\not\subseteq V_k$, which implies \textbf{T2}. 
Finally, by Theorem \ref{Hirsch}, every $(\tilde{f}, \tilde{g})\in \Delta(\delta, Q)$ is permanent and there is a uniform lower bound $\beta>0$ in the definition of permanence for all $(\tilde{f},\tilde{g})$. 

\section{Proof of Proposition \ref{LVprop}}\label{prop_proof}
For this proof, we use a measure theoretic framework for verifying the condition in our main theorem. We begin with some terminology.
A Borel probability measure $\mu$ on $\mathbb{R}_+^{n+m}$ is \emph{invariant} if $\mu(B)=\mu(B.t)$ for all $t>0$ and every Borel set $B\subset \mathbb{R}_+^{n+m}$.  An invariant measure $\mu$ is \emph{ergodic} if $\mu(B)=0$ or $1$ for any invariant Borel set $B$.  
Importantly, the following proposition refines the set on which the condition in our main theorem must hold.

\begin{proposition}\label{ergeq}
Let $E\subset S$ be a compact invariant set and $h: E\rightarrow \mathbb{R}$ a continuous function. Then, the following are equivalent:
\begin{enumerate}
\item For all $z \in E$, there is a $T_z>0$ such that $\int_0^{T_z} h(z(t))dt >0$\label{a}
\item $\int h(z) \mu(dz)>0$ for all invariant probability measures $\mu$ supported by $E$\label{b}
\item $\int h(z) \nu(dz)>0$ for all ergodic probability measures $\nu$ supported by $E$\label{c}
\end{enumerate}
\end{proposition}

The equivalence of (\ref{a}) and (\ref{b}) in Proposition \ref{ergeq} follows directly from Lemmas 4.1 and 4.2 in \citet{Garay2003}.  (\ref{b})$\Longrightarrow$ (\ref{c}) follows as ergodic probability measures are invariant probability measures. (\ref{c})$\Longrightarrow$ (\ref{b}) follows from the Ergodic decomposition theorem. This proposition implies that the final condition of the GALF need only be checked on the smallest invariant set containing the support of all the ergodic measures.  

\begin{proof}[Proof of Proposition \ref{LVprop}]
Choose the trivial Morse decomposition $\{\Gamma\cap S_0\}$ and consider the set of all ergodic measures $\{\mu_1,\dots, \mu_m\}$ supported by $\{\Gamma\cap S_0\}$. An interior face of $\{\Gamma\cap S_0\}$ is the set $\{(x,y)|x_i>0 \Longleftrightarrow i\in I\}$ for some proper subset $I\subset \{1,\dots, n\}$.  Since each interior face is an invariant set, all ergodic measures are supported on interior faces. For each $\mu_k$, let $\eta_k$ be the support of $\mu_k$ and let $\sigma_k$ be the interior face so that $\eta_k\subset \sigma_k$.  Then, for some weight $p$,

\begin{equation}\label{ergproofeq}
\begin{aligned}
\int \sum_i p_i  f_i(x,y)d\mu_k(x,y)=&\lim_{t\rightarrow\infty}\frac{1}{t}\int_0^t \sum_i p_i f_i(x(s),y(s))ds\\
=&\lim_{t\rightarrow\infty} \frac{1}{t} \int_0^t \sum_i p_i (\sum_j a_{ij} x_j(s) + b_i(y(s)))ds\\
=&\sum_i p_i [\sum_j a_{ij} (\lim_{t\rightarrow\infty} \frac{1}{t} \int_0^t x_j(s)ds) + (\lim_{t\rightarrow\infty} \frac{1}{t} \int_0^t b_i(y(s))ds)]\\
\end{aligned}
\end{equation}
where $(x(t), y(t))=z.t$ with $z\in \eta_k$.  ~\citet{Hofbauer1998}[Theorem 5.2.3] implies that 
\[
\lim_{t\rightarrow\infty} \frac{1}{t} \int_0^t x_j(s)ds=x^*_j
\]
for a unique $x^*$ on $\sigma_k$ such that $\sum_j a_{ij} x_j^* = -\overline{b_i}$ whenever $x_i^*>0$.
Finally, the assumption implies that there exists a $\vec{p}$ such that (\ref{ergproofeq}) is greater than zero and Proposition \ref{ergeq} and Theorem \ref{maintheorem_pert} conclude the proof. 

\end{proof}

\bibliography{Bib.bib}

\end{document}